\documentclass[smallextended]{svjour3}
\usepackage{graphicx}
\usepackage[nocompress]{cite}
\usepackage[cmex10]{amsmath}
\usepackage{amsfonts}
\usepackage{array}
\usepackage[tight,normalsize,sf,SF]{subfigure}
\usepackage[caption=false,font=normalsize,labelfont=sf,textfont=sf]{subfig}
\usepackage{url}
\usepackage[linesnumbered]{algorithm2e}
\usepackage[center]{caption}

\newcommand{\assign}{\,\leftarrow\,}

\author{
Minming Li
    \thanks{Department of Computer Science,
    City University of Hong Kong,
    83 Tat Chee Avenue, Kowloon,
    Hong Kong. Email: \texttt{minmli@cs.cityu.edu.hk} }
\and Frances F. Yao
    \thanks{Institute for Interdisciplinary Information Sciences,
    Tsinghua University,
    Beijing,
    China. Email: \texttt{csfyao@cityu.edu.hk} }
\and
Hao Yuan
    \thanks{Department of Computer Science,
    City University of Hong Kong,
    83 Tat Chee Avenue, Kowloon,
    Hong Kong. Email: \texttt{haoyuan@cityu.edu.hk} }
}

\begin{document}

\title{An $O(n^2)$ Algorithm for Computing Optimal Continuous Voltage Schedules}
\maketitle
%
%
%
%


\begin{abstract}
Dynamic Voltage Scaling techniques allow the processor to set its
speed dynamically in order to reduce energy consumption. In the
continuous model, the processor can run at any speed, while in the
discrete model, the processor can only run at finite number of
speeds given as input. The current best algorithm for computing the
optimal schedules for the continuous model runs at $O(n^2\log n)$
time for scheduling $n$ jobs. In this paper, we improve the running
time to $O(n^2)$ by speeding up the calculation of s-schedules using
a more refined data structure. For the discrete model, we improve
the computation of the optimal schedule from the current best
$O(dn\log n)$ to $O(n\log \max\{d,n\})$ where $d$ is the number of
allowed speeds.
\end{abstract}



\section{Introduction}
Energy efficiency is always a primary concern for chip designers not
only for the sake of prolonging the lifetime of batteries which are
the major power supply of portable electronic devices but also for
the environmental protection purpose when large facilities like data
centers are involved. Currently, processors capable of operating at
a range of frequencies are already available, such as Intel's
SpeedStep technology and AMD's PowerNow technology. The capability
of the processor to change voltages is often referred to in the
literature as DVS (Dynamic Voltage Scaling) techniques. For DVS
processors, since energy consumption is at least a quadratic
function of the supply voltage (which is proportional to CPU speed),
it saves energy to let the processor run at the lowest possible
speed while still satisfying all the timing constraints, rather than
running at full speed and then switching to idle.

One of the earliest theoretical models for DVS was introduced by
Yao, Demers and Shenker \cite{Yao95} in 1995. They assumed that the
processor can run at any speed and each job has an arrival time and
a deadline. They gave a characterization of the minimum-energy
schedule (MES) and an $O(n^3)$ algorithm for computing it which is
later improved to $O(n^2\log n)$ by \cite{Li06}. No special
assumption was made on the power consumption function except
convexity. Several online heuristics were also considered including
the Average Rate Heuristic (AVR) and Optimal Available Heuristic
(OPA). Under the common assumption of power function
$P(s)=s^{\alpha}$, they showed that AVR has a competitive ratio of
$2^{\alpha-1}\alpha^{\alpha}$ for all job sets. Thus its energy
consumption is at most a constant times the minimum required. Later
on, under various related models and assumptions, more algorithms
for energy-efficient scheduling have been proposed.

Bansal et al. \cite{Bansal04} further investigated the online
heuristics for the model proposed by \cite{Yao95} and proved that
the heuristic OPA has a tight competitive ratio of $\alpha^{\alpha}$
for all job sets. For the temperature model where the temperature of
the processor is not allowed to exceed a certain thermal threshold,
they showed how to solve it within any error bound in polynomial
time. Recently, Bansal et al. \cite{Bansal08} showed that the
competitive analysis of AVR heuristic given in \cite{Yao95} is
essentially tight. Quan and Hu \cite{Quan01} considered scheduling
jobs with fixed priorities and characterized the optimal schedule
through transformations to MES \cite{Yao95}. Yun and Kim
\cite{Yun03} later on showed the NP-hardness to compute the optimal
schedule.

Pruhs et al. \cite{Pruhs04} studied the problem of minimizing the
average flow time of a sequence of jobs when a fixed amount of
energy is available and gave a polynomial time offline algorithm for
unit-size jobs. Bunde \cite{Bunde06} extended this problem to the
multiprocessor scenario and gave some nice results for unit-size
jobs. Chan et al. \cite{Soda07} investigated a slightly more
realistic model where the maximum speed is bounded. They proposed an
online algorithm which is $O(1)$-competitive in both energy
consumption and throughput. More work on the speed bounded model can
be found in \cite{ICALP08}\cite{TAMC07}\cite{ISAAC07}.

Ishihara and Yasuura \cite{Ishihara98} initiated the research on
discrete DVS problem where a CPU can only run at a set of given
speeds. They solved the case when the processor is only allowed to
run at two different speeds. Kwon and Kim \cite{Kwon03} extended it
to the general discrete DVS model where the processor is allowed to
run at speeds chosen from a finite speed set. They gave an $O(n^3)$
algorithm for this problem based on the MES algorithm in
\cite{Yao95}, which is later improved in \cite{Li05} to $O(dn\log
n)$ where $d$ is the allowed number of speeds.

When the CPU can only change speed gradually instead of instantly,
\cite{Qu98} discussed about some special cases that can be solved
optimally in polynomial time. Later, Wu et al. \cite{Wu09} extended
the polynomial solvability to jobs with agreeable deadlines. Irani
et al. \cite{Irani03} investigated an extended scenario where the
processor can be put into a low-power sleep state when idle. A
certain amount of energy is needed when the processor changes from
the sleep state to the active state. The technique of switching
processors from idle to sleep and back to idle is called Dynamic
Power Management (DPM) which is the other major technique for energy
efficiency. They gave an offline algorithm that achieves
2-approximation and online algorithms with constant competitive
ratios. Recently, Albers and Antoniadis \cite{SODA12} proved the
NP-hardness of the above problem and also showed some lower bounds
of the approximation ratio. Pruhs et al. \cite{Pruhs10} introduced
profit into DVS scheduling. They assume that the profit obtained
from a job is a function on its finishing time and on the other hand
money needs to be paid to buy energy to execute jobs. They give a
lower bound on how good an online algorithm can be and also give a
constant competitive ratio online algorithm in the resource
augmentation setting. A survey on algorithmic problems in power
management for DVS by Irani and Pruhs can be found in
\cite{Irani05}. Most recent surveys by Albers can be found in
\cite{Albers10}\cite{Albers11b}.

In \cite{LiB06}, the authors showed that the optimal schedule for
tree structured jobs can be computed in $O(n^2)$ time. In this
paper, we prove that the optimal schedule for general jobs can also
be computed in $O(n^2)$ time, improving upon the previously best
known $O(n^2\log n)$ result \cite{Li06}. The remaining paper is organized
as follows. Section 2 will give the problem formulation. Section 3
will discuss the linear implementation of an important tool
--- the s-schedule used in the algorithm in \cite{Li06}. Then we use the
linear implementation to improve the calculation of the optimal
schedule in Section 4. In Section 5, we give improvements in the
computation complexity of the optimal schedule for the discrete
model. Finally, we conclude the paper in Section 6.

\section{Models and Preliminaries}

We consider the single processor setting. A job set
$J=\{j_1,j_2,\ldots,j_n\}$ over $[0,1]$ is given where each job
$j_k$ is characterized by three parameters: arrival time $a_k$,
deadline $b_k$, and workload $R_k$. Here workload means the required
number of CPU cycles. We also refer to $[a_k, b_k]\subseteq [0,1]$
as the interval of $j_k$. A {\it schedule} $S$ for $J$ is a pair of
functions $(s(t), job(t))$ which defines the processor speed and the
job being executed at time $t$ respectively. Both functions are
assumed to be piecewise continuous with finitely many
discontinuities. A {\it feasible} schedule must give each job its
required workload between its arrival time and deadline with perhaps
intermittent execution. We assume that the power $P$, or energy
consumed per unit time, is $P(s)=s^{\alpha}$ ($\alpha\geq 2$) where
$s$ is the processor speed. The total energy consumed by a schedule
$S$ is $E(S)=\int_0^1 P(s(t))dt$. The goal of the min-energy
feasibility scheduling problem is to find a feasible schedule that
minimizes $E(S)$ for any given job set $J$. We refer to this problem
as the continuous {\it DVS} scheduling problem.

For the continuous DVS scheduling problem, the optimal schedule
$S_{opt}$ is characterized by using the notion of a \textit{critical
interval} for $J$, which is an interval $I$ in which a group of jobs
must be scheduled at maximum constant speed $g(I)$ in any optimal
schedule for $J$. The algorithm MES in \cite{Yao95} proceeds by
identifying such a critical interval $I$, scheduling those \lq
critical' jobs at speed $g(I)$ over $I$, then constructing a
subproblem for the remaining jobs and solving it recursively. The
details are given below.

\begin{definition}
For any interval $I \subseteq [0,1]$, we use $J_{I}$ to denote the
subset of jobs in $J$ whose intervals are completely contained in
$I$. The intensity of an interval $I$ is defined to be
$g(I)=(\sum_{j_k\in J_{I}} R_k)/|I|$.
\end{definition}

An interval $I^\ast$ achieving maximum $g(I)$ over all possible
intervals $I$ defines a critical interval for the current job set.
It is known that the subset of jobs $J_{I^\ast}$ can be feasibly
scheduled at speed $g(I^\ast)$ over $I^\ast$ by the earliest
deadline first (EDF) principle. That is, at any time $t$, a job
which is waiting to be executed and having earliest deadline will be
executed during $[t,t+\epsilon]$. The interval $I^\ast$ is then
removed from $[0,1]$; all the remaining job intervals $[a_k, b_k]$
are updated to reflect the removal, and the algorithm recurses. We
denote the optimal schedule which guarantees feasibility and
consumes minimum energy in the continuous DVS model as \textit{OPT}.

The authors in \cite{Li06} later observed that in fact the critical
intervals do not need to be located one after another. Instead, one
can use a concept called $s$-schedule defined below to do
bipartition on jobs which gradually approaches the optimal speed
curve.

\begin{definition}
For any constant $s$, the $s$-schedule for $J$ is an EDF schedule
which uses a constant speed $s$ in executing any jobs of $J$. It will
give up a job when the deadline of the job has passed. In general,
$s$-schedules may have idle periods or unfinished jobs.
\end{definition}

\begin{definition}
In a schedule $S$, a maximal subinterval of $[0,1]$ devoted to
executing the same job $j_k$ is called an execution interval for
$j_k$ (with respect to $S$). Denote by $I_k(S)$ the union of all
execution intervals for $j_k$ with respect to $S$. Execution
intervals with respect to the $s$-schedule will be called
$s$-execution intervals.
\end{definition}

It is easy to see that the $s$-schedule for $n$ jobs contains at
most $2n$ $s$-execution intervals, since the end of each execution
interval (including an idle interval) corresponds to the moment when
either a job is finished or a new job arrives. Also, the
$s$-schedule can be computed in $O(n\log n)$ time by using a
priority queue to keep all jobs currently available, prioritized by
their deadlines. In the next section, we will show that the
$s$-schedule can be computed in linear time.

\section{Computing an $s$-Schedule in Linear Time}

In this work, we assume that the underlying computational model is the unit-cost RAM model with word size $\Theta(\log n)$. This model is assumed only for the purpose of using a special union-find algorithm by Gabow and Tarjan \cite{Gabow1983}.

\begin{theorem} \label{lemma:s-schedule}
If for each $k$, the rank of $a_k$ in $\{ a_1, a_2, \ldots, a_n\}$ and the rank of $b_k$ in $\{b_1, b_2, \ldots, b_n \}$ are pre-computed, then the $s$-schedule can be computed in linear time in the unit-cost RAM model.
\end{theorem}

We make the following two assumptions:
\begin{itemize}
\item the jobs are already sorted according to their deadlines;
\item for each job $j_k$, we know the rank of $a_k$ in the arrival time set $\{ a_1, a_2, \ldots, a_n  \}$.
\end{itemize}
Because of the first assumption and without loss of generality, we assume that  $b_1 \le b_2 \le \ldots \le b_n$.
Algorithm \ref{algo:naive} schedules the jobs in the order of their deadlines. When scheduling job $k$, the algorithm tries to search for an earliest available time interval and schedule the job in it, and then repeat the process until all the workload of the job is scheduled or unable to find such a time interval before the deadline. A more detailed discussion of the algorithm is given below.


\begin{algorithm}[t]
\caption{Algorithm for Computing an $s$-Schedule }
\label{algo:naive} Initialize $e_i \assign t_i$ for $1 \le i < m$.
\nllabel{line:initE} \; \For{k=1 \KwTo n} {
    Let $i$ be the rank of $a_k$ in $T$, i.e., $t_i=a_k$. \nllabel{line:set_i} \;
    Initialize $r \assign R_k$, where $r$ denotes the remaining workload to be scheduled. \;
    \While{$r>0$ \nllabel{line:beginwhile} }
    {
        Search for an earliest \textit{non-empty} canonical time interval $[e_{p}, t_{p+1})$ such that $e_p \ge t_i$.  \nllabel{line:findEarliest} \;

        \If
        {
            $e_p \ge b_k$ \nllabel{line:deadlineCheck:begin}
        }  
        {
            Break the while loop because the job cannot be finished.
        }
        \nllabel{line:deadlineCheck:end}

        Set $u \assign \min\{  b_k, t_{p+1} \}$. \nllabel{line:setU}

        \eIf{$r > s \cdot (u - e_{p})$}
        {
            Schedule job $k$ at $[e_{p}, u)$. \nllabel{line:schedule1} \;
            Update $e_{p} \assign u$.  \nllabel{line:updateE1}  \;
            Update $r \assign r -  s \cdot (u - e_{p})$. \nllabel{line:updateR1}  \;
        }
        {
            Schedule job $k$ at $[e_{p}, e_{p}+r/s)$. \nllabel{line:schedule2} \;
            Update $e_{p} \assign e_{p}+r/s$.             \nllabel{line:updateE2} \;
            Update $r \assign 0$.
            \nllabel{line:updateR2}
        }
    } \nllabel{line:endwhile}
}
\end{algorithm}

Let $T$ be $\{a_1, a_2, \ldots, a_n, 1, 1+\epsilon\}$. Note that the
times ``$1$'' and ``$1+\epsilon$'' (where $\epsilon$ is any fixed
positive constant) are included in $T$ for simplifying the
presentation of the algorithm. Denote the size of $T$ by $m$. Denote
$t_i$ to be the $i$-th smallest element in $T$. Note that the rank of any $a_k$ in $T$ is known. During the running
of the algorithm, we will maintain the following data structure:
\begin{definition}
For each $1\le i < m$, the algorithm maintains a value $e_i$,
 whose value is in the range $[t_i, t_{i+1}]$. The meaning of $e_i$ is that: the time
interval $[t_i, e_i)$ is fully occupied by some jobs, and the time
interval $[e_i, t_{i+1})$ is idle.
\end{definition}
If $[t_i, t_{i+1})$ is fully
occupied, then $e_i$ is $t_{i+1}$. Note that such a time $e_i$
always exists during the running of the algorithm, which will be
shown later when we discuss how to maintain $e_i$. At the beginning
of the algorithm, we assume that the processor is idle for the whole
time period. That means $e_i=t_i$ for $1 \le i < m$ (see line
\ref{line:initE} of Algorithm \ref{algo:naive}).

\begin{example}
An example for demonstrating the usage of the $e_i$ data structure is given below: Assume that $T=\{0.1, 0.2, 0.3, 0.4, 0.5, 0.6, 0.7, 0.8, 0.9, 1, 1+\epsilon \}$. At some point during the execution of the algorithm, if some jobs have been scheduled to run at time intervals $[0.2, 0.35), [0.6, 0.86), [0.9, 0.92)$, then we will have $e_1 = 0.1$, $e_2 = 0.3$, $e_3 = 0.35$, $e_4=0.4$, $e_5=0.5$, $e_6=0.7$, $e_7=0.8$, $e_8=0.86$, $e_9=0.92$, and $e_{10}=1$.
\end{example}

\begin{figure*}
\centering
\includegraphics[bb = 150 600 550 800]{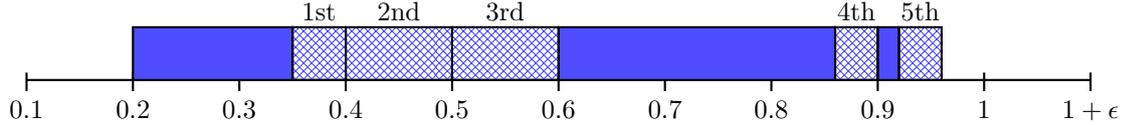}
\caption{An illustration for example 2.} \label{figure1}
\end{figure*}

Before we analyze the algorithm, we need to define an important concept called \textrm{canonical time interval}.
\begin{definition}
During the running of the algorithm, a canonical time interval is a time interval of the form $[e_{p}, t_{p+1})$, where $1 \le p <m$. When  $e_p=t_{p+1}$, we call it an empty canonical time interval.
\end{definition}
Note that a non-empty canonical time interval is always idle based on the definition of $e_p$.
Any arrival time $a_k$ will not lie inside any canonical time interval but it is possible that $a_k$ will touch any of the two ending points, i.e., for any $1 \le p < m$, we have either $a_k \le e_p$ or $a_k \ge t_{p+1}$.
Therefore, if we want to search for a time interval to run a job at or after time $a_k$, then we should always look for the earliest non-empty canonical time interval $[e_p, t_{p+1})$ where $e_p \ge a_k$.

In Algorithm \ref{algo:naive}, a variable $r$ is used to track the workload to be scheduled. Lines \ref{line:beginwhile}-\ref{line:endwhile} try to schedule $j_k$  as early as possible if $r>0$. Line \ref{line:findEarliest} tries to search for an earliest \textrm{non-empty} canonical time interval $[e_{p}, t_{p+1})$ no earlier than the arrival time of $j_k$ (i.e., $e_p \ge a_k$). Such a $p$ always exists because there is always a non-empty canonical time interval $[1, 1+\epsilon)$.
Line \ref{line:deadlineCheck:begin}-\ref{line:deadlineCheck:end} means that, if $e_p$ is not earlier than the deadline of $j_k$, then the job cannot be finished. Line \ref{line:setU} sets a value of $u$, whose meaning is that $[e_p, u)$ can be used to schedule the job. The value of $u$ is no later than the deadline of $j_k$.
Lines \ref{line:schedule1}-\ref{line:updateR1} process the case when the remaining workload of $j_k$ cannot be finished in the time interval $[e_p,u)$.
Lines \ref{line:schedule2}-\ref{line:updateR2} process the case when the remaining workload of $j_k$ can be finished in the time interval $[e_p, u)$.
In the first case, line \ref{line:updateE1} updates $e_{p}$ to $u$ because the time interval $[t_{p}, u)$ is occupied and $[u,t_{p+1})$ is idle. In the second case, a time of $r/s$ is occupied by $j_k$ after the time $e_{p}$, so $e_{p}$ is increased by $r/s$.

\begin{example}
Following the example provided in the previous example, assume that the speed is $s=1$, if we are to schedule a job $j_k$, where $a_k=0.3, b_k=0.96, R_k=0.35$, the algorithm will proceed as follows: At the beginning, $r$ will be initialized to $0.35$, and $i=3$ (because $a_k=0.3=t_3$; see line \ref{line:set_i}). Line \ref{line:findEarliest} will then get the interval $[e_3, t_4)=[0.35,0.4)$ as an earliest non-empty canonical time interval, and a workload of $(0.4-0.35)s=0.05$ is scheduled at that time interval. The values of $e_3$ will be updated to $0.4$ accordingly.
Now, $r$ becomes $0.35-0.05=0.3$, and line \ref{line:findEarliest} will get the time interval $[e_4,t_5)=[0.4,0.5)$ to schedule the job. After that $r$ becomes $0.3-(0.5-0.4)s=0.2$, and $e_4=0.5$. Line \ref{line:findEarliest} then gets the time interval $[e_5,t_6)=[0.5,0.6)$ to schedule the job, and $r$ will be further reduced to $0.1$. The values of $e_5$ will be updated to $0.6$. The next time interval found will be $[e_8,t_9)=[0.86,0.9)$, and $r$ will become $0.1-(0.9-0.86)s=0.06$. The values of $e_8$ will be updated to $0.9$. The remaining earliest non-empty canonical time interval is $[e_9, t_{10})=[0.92, 1)$, but the deadline of the job is $0.96$, so only $[0.92,0.96)$ will be used to schedule the job, and $r$ will be $0.02$. The value of $e_9$ is then updated to $0.96$.
Finally, $[e_9, t_{10})=[0.96,1)$ is the remaining earliest non-empty canonical time interval, but $e_9 \ge b_k$, so line \ref{line:deadlineCheck:begin}-\ref{line:deadlineCheck:end} will break the loop, and $j_k$ will be an unfinished job. A graphical illustration is provided in Figure \ref{figure1}. The solid rectangles represent the time intervals occupied by some jobs before scheduling $j_k$. The cross-hatched rectangles represent the time intervals that are used to schedule $j_k$. The $q$-th cross-hatched rectangle (where $1 \le q \le 5$) is the $q$-th time interval scheduled according to this example. Note that all the cross-hatched rectangles except the $5$-th one are canonical time intervals right before scheduling $j_k$.
\end{example}

The most critical part of the algorithm is Line \ref{line:findEarliest}, which can be implemented efficiently by the following folklore method using a special union-find algorithm developed by Gabow and Tarjan \cite{Gabow1983} (see also the discussion of the decremental marked ancestor problem \cite{AlstrupHR1998}). At the beginning, there is a set $\{ i \} $ for each $1 \le i < m$. The name of a set is the largest element of the set.
Whenever $e_p$ is updated to $t_{p+1}$ (i.e., there is not any idle time in the interval $[t_p, t_{p+1})$), we make a union of the set containing $p$ and  the set containing $p+1$, and set the name of this set to be the name of the set containing $p+1$. After the union, the two old sets are destroyed. In this way, a set is always an interval of integers.
For a set whose elements are $\{q,q+1,\ldots, p\}$, the semantic meaning is that, $[t_q, e_p)$ is fully scheduled but $[e_p, t_{p+1})$ is idle. Therefore, to search for an earliest non-empty canonical time interval beginning at or after time $t_i$, we can find the set containing $i$, and let $p$ be the name of the set, then $[e_p, t_{p+1})$ is the required time interval.

\begin{example}
An example of the above union-find process for scheduling $j_k$ in the previous example is given below: Before scheduling $j_k$, we have the sets $\{1\}$, $\{2,3\}$, $\{4\}$, $\{5\}$, $\{6,7,8\}$, $\{9\}$, $\{10\}$.
The execution of line \ref{line:findEarliest} will always try to search for a set that contains the element $i=3$. Therefore, the first execution will find the set $\{2,3\}$, so $p$ will be $3$. After that, $e_3$ becomes $t_4=0.4$, so the algorithm needs to make a union of the sets $\{2,3\}$ and $\{4\}$ to get $\{2,3,4\}$. Similarly, the next execution will find the set $\{2,3,4\}$, so $p=4$. The algorithm will then make a union of $\{2,3,4\}$ and $\{5\}$ to get $\{2,3,4,5\}$. For the next execution, the set $\{2,3,4,5\}$ will be found, and it will be merged with $\{6,7,8\}$ to get $\{2,3,4,5,6,7,8\}$. In this case, $p=8$, and the earliest non-empty canonical time interval is $[e_p,t_{p+1})=[0.86,0.9)$. After $e_8$ is updated to $t_9=0.9$, the algorithm will merge $\{2,3,4,5,6,7,8\}$ with $\{9\}$ and obtain $\{2,3,4,5,6,7,8,9\}$. Therefore, the next execution of line \ref{line:findEarliest} will get $p=9$. After the time interval $[0.92,0.96)$ is scheduled and $e_9$ is updated to $0.96$, so the algorithm will not do any union. The last execution finds $p=9$ again, and a loop break is performed.
\end{example}

Now, we we analyze the time complexity of the algorithm.
\begin{lemma}
Each set always contains continuous integers.
\end{lemma}
\begin{proof}
It can be proved by induction. At the beginning, each skeleton set is a continuous integer set. During the running of the algorithm, the union operation always merges two nearby continuous integer sets to form a larger continuous integer set.
\end{proof}
\begin{lemma}
There are at most $m-2$ unions.
\end{lemma}
\begin{proof}
It is because there are only $m-1$ sets.
\end{proof}

\begin{lemma}
There are at most $2(m-2)+n$ finds.
\end{lemma}
\begin{proof}
Some $m-2$ finds are from finding the set containing $p+1$ during each union. Note that there is no need to perform a find operation to find the set containing $p$ for union, because $p$ is just the name of such a set, where the set contains continuous integers with $p$ as the largest element. The other $(m-2)+n$ finds are from searching for earliest canonical time intervals beginning at or after time $t_i$. This can be analyzed in the following way: Let $z_k$ be the number of times to search for an earliest \textrm{non-empty} canonical time interval when processing job $j_k$. Let $w_k$ be the number of unions that are performed when processing job $j_k$. We have $z_k \le w_k+1$, because each of the first $z_k-1$ finds must accompany a union. Therefore, $$\sum_{1 \le k \le n} z_k \le \sum_{1 \le k \le n} (w_k+1) = \sum_{1 \le k \le n} w_k + n \le (m-2)+n.$$
\end{proof}

Since these unions and finds are operated on the sets of integer intervals,
such an interval union-find problem can be solved in $O(m+n)$ time in the unit-cost RAM model using Gabow and Tarjan's algorithm \cite{Gabow1983}. Note that $m = O(n)$, so
the total time complexity is $O(n)$. Theorem \ref{lemma:s-schedule} holds.

If the union-find algorithm is implemented in the pointer machine model \cite{BenAmram1995} using the classical algorithm of Tarjan \cite{Tarjan1975}, the complexity of our $s$-schedule algorithm will become $O(n \alpha(n))$ where $\alpha(n)$ is the one-parameter inverse Ackermann function.

Note that, the number of finds can be further reduced with a more careful implementation of the algorithm as follows (but the asymptotic complexity will not change):
\begin{itemize}
\item Whenever the algorithm schedules a job $j_k$ to run at a time interval $[e_p, b_k)$, the algorithm no longer needs to proceed to line \ref{line:findEarliest} for the same job, because there will not be any idle time interval available before the deadline.

\item For each job $j_k$, the first time to find a non-empty canonical time interval requires one find operation. In any of the later times to search for earliest non-empty canonical time intervals for the same job, there must be a union operation just performed. The $p$ that determines the earliest non-empty canonical time interval $[e_p,t_{p+1})$ is just the name of that new set after that union, so a find operation is not necessary in this case. Note that the find operations that accompany the unions are still required.
\end{itemize}
Using the above implementation, the number of finds to search for earliest non-empty canonical time intervals can be reduced to $n$. Along with the $m-2$ finds for unions, the total number of finds of this improved implementation is at most $(m-2)+n$.

\section{An $O(n^2)$ Continuous DVS Algorithm}
We will first take a brief look at the previous best known DVS algorithm of Li, Yao and Yao \cite{Li06}. As in \cite{Li06}, Define the ``support'' $U$ of $J$ to be the union
of all job intervals in $J$. Define $\textrm{avr}(J)$, the ``average rate''
of $J$ to be the total workload of $J$ divided by $|U|$. According
to Lemma 9 in \cite{Li06}, using $s=\textrm{avr}(J)$ to do an $s$-schedule will
generate two nonempty subsets of jobs requiring speed at least $s$
or less than $s$ respectively in the optimal schedule unless the
optimal speed for $J$ is a constant $s$. The algorithm will recursively do the scheduling based on the two subsets of jobs. Therefore, at most $n$
calls of $s$-schedules on a job set with at most $n$ jobs are needed
before we obtain the optimal schedule for the whole job set. The most time-consuming part of their algorithm is the $s$-schedules.

To apply our improved $s$-schedule algorithm for solving the continuous DVS scheduling problem, we need to make sure that the ranks of the deadlines and arrival times are known before each $s$-schedule call.
It can be done in the following way:
Before the first call, sort the deadlines and arrival times and obtain the ranks.
In each of the subsequent calls, in order to get the new ranks within the two subsets of jobs, a counting sort algorithm can be used to sort the old ranks in linear time. Therefore, the time to obtain the ranks is at most $O(n^2)$ for the whole algorithm. Based on the improved computation of $s$-schedules, the total time complexity of the DVS problem is now $O(n^2)$, improving the previous $O(n^2 \log n)$ algorithm of \cite{Li06} by a factor of $O(\log n)$. We have the following theorem.
\begin{theorem}
The continuous DVS scheduling problem can be solved in $O(n^2)$ time for $n$ jobs in the unit-cost RAM model.
\end{theorem}

\section{Further Improvements}
For the discrete DVS scheduling problem, we have an $O(n\log n)$
algorithm to calculate the optimal schedule by doing binary testing
on the given $d$ speed levels, improving upon the previously best
known $O(dn\log n)$ \cite{Li05}. To be specific, given the input job
set with size $n$ and a set of speeds $\{s_1, s_2, \ldots, s_d\}$,
we first choose the speed $s_{d/2}$ to bi-partition the job set into
two subsets. Then within each subset, we again choose the middle
speed level to do the bi-partition. We recursively do the
bi-partition until all the speed levels are handled. In the
recursion tree thus built, we claim that the re-sorting for
subproblems on the same level can be done in $O(n)$ time which
implies that the total time needed is $O(n\log d+n\log n)=O(n\log
\max\{d,n\})$. The claim can be shown in the following way. Based on
the initial sorting, we can assign a new label to each job
specifying which subgroup it belongs to when doing bi-partitioning.
Then a linear scan can produce the sorted list for each subgroup.


\section{Conclusion}
In this paper, we improve the time for computing the optimal
continuous DVS schedule from $O(n^2\log n)$ to $O(n^2)$. The major
improvement happens in the computation of s-schedules. Originally,
the s-schedule computation is done in an online fashion where the
execution time is allocated from the beginning to the end
sequentially and the time assigned to a certain job can be gradually
decided. While in this work, we allocate execution time to jobs in
an offline fashion. When jobs are sorted by deadlines, job $j_i$'s
execution time is totally decided before we go on to consider
$j_{i+1}$. Then by using a suitable data structure and conducting a
careful analysis, the computation time for s-schedules improves from
$O(n\log n)$ to $O(n)$. We also design an algorithm to improve the
computation of the optimal schedule for the discrete model from
$O(dn\log n)$ to $O(n\log \max\{d,n\})$.


\end{document}